\DeclareFontFamily{TU}{ptm}{}\DeclareFontShape{TU}{ptm}{m}{n}{<->ssub*lmr/m/n}{}\DeclareFontShape{TU}{ptm}{bx}{n}{<->ssub*lmr/m/n}{}\DeclareFontShape{TU}{ptm}{m}{it}{<->ssub*lmr/m/n}{}\DeclareFontShape{TU}{ptm}{bx}{it}{<->ssub*lmr/m/n}{}
\par
\documentclass[journal,twoside,web]{arxiv-color} \usepackage{generic,lcsys} \par
\usepackage[T1]{fontenc} \usepackage{amsmath,amssymb,amsfonts} \usepackage{graphicx,booktabs,cite,algorithm,algorithmic} \usepackage[hidelinks]{hyperref} \newtheorem{theorem}{Theorem} \newtheorem{lemma}{Lemma} \newtheorem{proposition}{Proposition} \newtheorem{assumption}{Assumption} \newcommand{\aftertheoremdisplay}{\par
\addvspace{8pt plus 2pt minus 1pt}} \DeclareMathOperator{\tr}{tr} \DeclareMathOperator{\rank}{rank} \renewenvironment{proof}{\par
\noindent\textit{Proof:}\ }{\unskip\nobreak\hfill\penalty50\hskip1em\hbox{}\nobreak\hfill$\square$\par
}  \title{Globally Certified Invariant-Ellipsoid\\Control from Data} \author{Ngoc~Tuan~Dinh, Egor~Dogadin, and Alexey~Peregudin\thanks{Ngoc Tuan Dinh is with ITMO University, St. Petersburg, Russia (email: dinhngoctuan6789@gmail.com).}\thanks{Egor Dogadin is with ITMO University, St. Petersburg, Russia (email: egor.dogadin@icloud.com).}\thanks{Alexey Peregudin is with the School of Electrical and Electronic Engineering, University of Sheffield, Sheffield, United Kingdom (e-mail: a.peregudin@sheffield.ac.uk).}\endgraf} \hypersetup{pdftitle={Globally Certified Invariant-Ellipsoid Control from Data},pdfauthor={Ngoc Tuan Dinh, Egor Dogadin, Alexey Peregudin}} \AtBeginDvi{}\hbadness=10000\vbadness=10000
\begin{document} \maketitle\thispagestyle{empty} \begin{abstract} This letter develops a data-based method for designing state feedback for a discrete-time linear system under bounded disturbances. For each admissible feedback gain and scalar design parameter, a Lyapunov equation determines an invariant ellipsoid: a region that the state cannot leave under the permitted disturbances. We optimise the gain and parameter to minimise a trace-based measure of the resulting output enclosure. At each parameter value, value iteration gives a lower cost bound, while separate controller evaluation gives an achievable upper bound. These bounds also let us assess whole parameter intervals without evaluating every point. We prove that the search stops after finitely many evaluations with a stabilising state-feedback gain whose objective is within any prescribed absolute tolerance of the infimum over the chosen ellipsoid family. Neither an initially stabilising gain nor attainment of the infimum is assumed. The guarantee assumes exact arithmetic and a sufficiently informative batch of exact measurements, including disturbances during data collection; the resulting feedback uses only the state. A position--velocity example illustrates the bounds, controller checks, and computational cost. \par
\end{abstract} \par
 \par
\section{Introduction} Persistent disturbances can have a known amplitude bound without a useful stochastic model or finite energy. A controller can then be assessed by a set enclosing its possible responses. Invariant ellipsoids provide tractable geometric bounds through Lyapunov equations or inequalities, as in classical LMI-based design~\cite{Boyd1994,Nazin2007,Khlebnikov2011}. Recent duality results yield Riccati equations for the corresponding trace criterion in continuous time~\cite{Peregudin2024} and discrete time~\cite{Dogadin2024}. Each Riccati solution is conditional on a scalar ellipsoid parameter, leaving a further optimisation problem. \par
Direct data-based formulas provide stabilising and optimal feedback from informative trajectories~\cite{DePersis2020}; informativity characterises when a batch determines a control property~\cite{vanWaarde2020}. Sets of compatible models can represent bounded unmeasured errors~\cite{Bisoffi2022}. Data-based invariant-set synthesis also includes polytopic constructions that maximise an admissible region under state and input constraints~\cite{Mejari2023}. Here, we minimise the output trace of a specified invariant-ellipsoid family using exact measurements of the disturbances present during collection. \par
Bellman and Q-learning methods provide another route. Efficient offline Q-learning for discrete-time LQR includes data-based initial stabilising policies and comparisons with identify-then-design~\cite{Lopez2023}. Value iteration can avoid an initially admissible policy in continuous-time adaptive optimal control~\cite{Bian2016}; Q-learning also addresses completely unknown discrete-time LQR systems~\cite{Fan2025}. For invariant ellipsoids, Dinh et al.~\cite{Dinh2025} combine continuous-time policy iteration with parameter adaptation. We retain the discrete-time trace criterion of~\cite{Dogadin2024} and ask when computation can stop with a certified error relative to its full-domain infimum. \par
Asymptotic convergence alone does not answer this question. A finite-horizon critic is a lower value bound; its greedy policy can perform worse or be unstable. General value-iteration stopping theory already connects finite computation, stability, and near-optimality~\cite{Granzotto2021,Granzotto2026}. The additional issue here is to certify the returned invariant ellipsoid while accounting for every unexplored parameter interval, including a possible unattained boundary optimum. \par
Our main result is a finite global certificate. Exact data collected under measured disturbances are used to reconstruct the Bellman operator through deterministic virtual probes, and policy evaluation supplies admissibility and objective bounds. Monotonicity of the normalised value matrix yields interval lower bounds and a computable upper-tail exclusion. An explicit interval-width estimate then proves that best-first bisection terminates for any absolute tolerance, without an interior-minimiser assumption. The result is a stabilising gain, its data-derived invariant ellipsoid, and a gap to the global infimum. The same outer certificate can be combined with a Riccati solver; matched numerical comparisons quantify the computational cost. \par
\section{Problem and Exact Data} \label{sec:problem} Consider the discrete-time linear system \begin{equation} x_{k+1}=Ax_k+Bu_k+Ew_k,\qquad z_k=Cx_k+Du_k, \label{eq:plant} \end{equation} where $x_k\in\mathbb R^n$, $u_k\in\mathbb R^m$, $w_k\in\mathbb R^r$, and $z_k\in\mathbb R^p$. Disturbances satisfy $\|w_k\|_2\leq1$ at every step. We seek a static feedback $u_k=Kx_k$. Write $F_K=A+BK$ and $C_K=C+DK$; $\rho(\cdot)$ denotes spectral radius, $\tr(\cdot)$ trace, and $\succeq$ the ordering of symmetric matrices. Vector norms are Euclidean. \par
\subsection{Which ellipsoid is optimised?} An invariant set contains the next state whenever it contains the current state, for every permitted disturbance. For $P\succeq0$, define \begin{equation} \mathcal E(P)=\{P^{1/2}v:\|v\|_2\leq1\}. \label{eq:ellipsoid} \end{equation} When $P\succ0$, this is the familiar ellipsoid $x^\top P^{-1}x\leq1$. A singular $P$ describes an ellipsoid in a lower-dimensional subspace. \par
For a gain and parameter satisfying $\rho(F_K)^2<\alpha<1$, the unique solution of \begin{equation} P=\alpha^{-1}F_KPF_K^\top+(1-\alpha)^{-1}EE^\top\label{eq:primal} \end{equation} is positive semidefinite and defines an invariant ellipsoid. We call such a pair $(\alpha,K)$ admissible. Thus $P$ is determined by the pair, rather than chosen independently; different admissible values of $\alpha$ can give different ellipsoids for the same gain. The output image is $C_K\mathcal E(P)$. Its squared semiaxis lengths sum to \begin{equation} J(\alpha,K)=\tr(C_KPC_K^\top). \label{eq:objective} \end{equation} Thus $J$ measures the size of a guaranteed output enclosure under persistent bounded disturbances. It also bounds $\|z_k\|_2^2$ for states in the ellipsoid. We optimise this trace within family~\eqref{eq:primal}; it is neither ellipsoid volume nor the exact induced peak gain. \par
Define \begin{equation} f(\alpha)=\min_{K:\rho(F_K)^2<\alpha}J(\alpha,K), \qquad J_\star=\inf_{0<\alpha<1}f(\alpha). \label{eq:optimum} \end{equation} The infimum is essential: for $A=0$, $B=E=1$, and $z=[x\ u]^\top$, $f(\alpha)=1/(1-\alpha)$ approaches $1$ only as $\alpha\downarrow0$. Given $\delta>0$, our goal is to return an admissible gain and its ellipsoid with $J(\widehat\alpha,\widehat K)-J_\star\leq\delta$, together with computable lower and upper bounds proving this gap. \par
\subsection{Information available to the computation} A finite batch of one-step measurements is arranged as $\mathcal D=(X,U,W,X^+,Z)$, with the samples as columns, so \begin{equation} X^+=AX+BU+EW,\qquad Z=CX+DU. \label{eq:data} \end{equation} Samples may come from a trajectory or several episodes. The disturbance is measured during collection, along with the state, input, and regulated output; it is not measured by the deployed feedback. No numerical plant matrices need be supplied. \par
\begin{assumption}\label{ass:main} The pair $(A,B)$ is controllable, $(C,A)$ is observable, $C^\top D=0$, $D^\top D\succ0$, and $E\ne0$. The measurements satisfy~\eqref{eq:data} exactly, and \begin{equation} M=\begin{bmatrix}X\\U\\W\end{bmatrix},\qquad\rank M=n+m+r. \label{eq:rank} \end{equation} \end{assumption} \aftertheoremdisplay Controllability and observability ensure that the scaled LQR problem below is well posed for every positive $\alpha$. The computation is offline; it requires neither an initial stabilising policy nor application of intermediate gains. Safe exploration and unmeasured learning errors are outside this information pattern. \par
Let $M^\dagger$ be the Moore--Penrose inverse. For $v\in\mathbb R^{n+m}$, form \begin{equation} g(v)=M^\dagger\begin{bmatrix}v\\0\end{bmatrix},\quad t(v)=X^+g(v),\quad o(v)=Zg(v). \label{eq:virtual} \end{equation} Because $Mg(v)=[v^\top\ 0]^\top$, the data equations imply $t(v)=[A\ B]v$ and $o(v)=[C\ D]v$. These are nominal virtual transitions: linear combinations remove the measured disturbance before the output is squared. Likewise, \begin{equation} \begingroup\setlength{\arraycolsep}{1pt} E\!=\!X^+M^\dagger\!\!\begin{bmatrix}0\\I_r\end{bmatrix}\!\!,\, F_K\!=\!X^+M^\dagger\!\!\begin{bmatrix}I_n\\K\\0\end{bmatrix}\!\!,\, C_K\!=\!ZM^\dagger\!\!\begin{bmatrix}I_n\\K\\0\end{bmatrix}\!\!. \endgroup\label{eq:maps} \end{equation} Thus the disturbance channel and each candidate closed-loop map are recovered explicitly. Assumption~\ref{ass:main} also permits exact identification of $A,B,E,C,D$. \par
To reconstruct a quadratic form, evaluate it on coordinate probes. If $h(v)=v^\top Hv$, the probes $e_i$ and $e_i+e_j$ give \begin{equation} H_{ii}=h(e_i),\quad H_{ij}=\tfrac12\bigl(h(e_i+e_j)-h(e_i)-h(e_j)\bigr). \label{eq:polarization} \end{equation} These deterministic probes reconstruct every symmetric $H$ and impose no additional quadratic-feature rank condition. They are synthetic combinations of the same batch, not extra measurements. The right inverse and virtual probes are computed once and reused for every parameter and iteration. \par
\section{A Finite Policy Certificate} \label{sec:policy} For a fixed ellipsoid parameter $\alpha$, value iteration provides lower bounds on the optimal cost, while separate evaluation of admissible controllers provides upper bounds. \par
\subsection{Building a lower cost bound} Value iteration builds the cost one time step at a time. Each update adds the current output penalty to the scaled estimate of future cost, then minimises over the input. Starting from zero gives lower values because a finite horizon omits nonnegative future penalties. \par
For $0<\alpha\leq1$ and $S\succeq0$, define the Bellman matrix by \begin{equation} v^\top H_\alpha(S)v=\|o(v)\|^2+\alpha^{-1}t(v)^\top S t(v). \label{eq:bellman} \end{equation} This quadratic state--action value is the Q-function for a scaled LQR problem. Partition $H=H_\alpha(S_j)$ by state and input coordinates. Starting from $S_0=0$, perform \begin{equation} K_j=-H_{uu}^{-1}H_{ux},\quad S_{j+1}=H_{xx}-H_{xu}H_{uu}^{-1}H_{ux}. \label{eq:vi} \end{equation} The input block is positive definite because $D^\top D\succ0$. The gain $K_j$ minimises the expression based on $S_j$; the updated value $S_{j+1}$ includes one more time step. \par
\begin{lemma}[Value bounds]\label{lem:vi} For each $0<\alpha<1$, iteration~\eqref{eq:vi} satisfies $0\preceq S_j\uparrow S_\star$, $K_j\to K_\star$, and \begin{equation} f(\alpha)=\frac{\tr(E^\top S_\star(\alpha)E)}{1-\alpha}, \qquad\rho(F_{K_\star})^2<\alpha. \label{eq:value} \end{equation} Moreover, $S_j(a)\succeq S_j(b)$ for $0<a\leq b\leq1$. \end{lemma} \begin{proof} The data identities give \begin{equation} H_\alpha(S)=[C\ D]^\top[C\ D] +\alpha^{-1}[A\ B]^\top S[A\ B]. \label{eq:Hmodel} \end{equation} Completing the square in the input shows that $x^\top S_jx$ is the $j$-stage optimal cost for dynamics $(A,B)/\sqrt\alpha$ and stage cost $\|Cx+Du\|^2$. Controllability and observability are preserved by this scaling. Standard finite-horizon LQR convergence~\cite{Lewis2012} therefore gives the stabilising infinite-horizon value and gain. Its policy equation is \begin{equation} S^K=C_K^\top C_K+\alpha^{-1}F_K^\top S^K F_K. \label{eq:policy} \end{equation} For any admissible $K$, substitution into~\eqref{eq:objective} and cyclicity of the trace give \begin{equation*} \begin{aligned} J(\alpha,K) &=\tr\!\left[S^K\!\left(P-\alpha^{-1}F_KPF_K^\top\right)\right]\\ &=\frac{\tr(E^\top S^K E)}{1-\alpha}, \end{aligned} \end{equation*} where the second equality uses~\eqref{eq:primal}. Applying this identity to $K_\star$ proves~\eqref{eq:value}. The identity connects geometry and cost: $P$ describes the state enclosure, whereas $S^K$ evaluates the controller. The disturbance channel $E$ enters through the trace after this nominal cost is computed. Finally, Bellman minimisation preserves matrix order in both $S\succeq0$ and $1/\alpha$. Induction from zero proves the parameter ordering. \end{proof} \par
\subsection{Evaluating and accepting a controller} The iterate $S_j$ gives an optimal finite-horizon value, not the cost of repeatedly applying its greedy gain $K_j$. We therefore evaluate that gain separately. The finite-horizon value supplies a lower bound on the optimum; an admissible policy supplies an upper bound. Their difference tells us when this fixed-parameter computation can stop. \par
\begin{proposition}[Accepted policies and their ellipsoids]\label{prop:policy} Fix $0<\alpha<1$. For a candidate $K_j$, solve~\eqref{eq:policy} using the data-derived $F_{K_j},C_{K_j}$. Accept the candidate only when the equation has a unique symmetric solution $S^{K_j}\succ0$. Every accepted candidate satisfies $\rho(F_{K_j})^2<\alpha$, and \begin{equation} L_j=\frac{\tr(E^\top S_jE)}{1-\alpha} \ \leq f(\alpha)\leq\
 U_j=\frac{\tr(E^\top S^{K_j}E)}{1-\alpha}. \label{eq:LU} \end{equation} The solution $P\succeq0$ of~\eqref{eq:primal} defines an invariant ellipsoid and satisfies $J(\alpha,K_j)=U_j$. For every $\eta>0$, an accepted candidate with $U_j-L_j\leq\eta$ is obtained after finitely many iterations. \end{proposition} \begin{proof} Suppose $F_Kv=\lambda v$ for a nonzero possibly complex $v$. Equation~\eqref{eq:policy} implies \begin{equation} (1-|\lambda|^2/\alpha)v^*S^Kv =\|Cv\|^2+\|DKv\|^2. \label{eq:gateproof} \end{equation} For $S^K\succ0$, $|\lambda|^2>\alpha$ is impossible. Equality forces $Cv=Kv=0$, hence $Av=\lambda v$, contradicting observability. Thus the policy is admissible. The lower bound follows from Lemma~\ref{lem:vi}. The upper bound follows from the definition of $f(\alpha)$, and $J(\alpha,K_j)=U_j$ follows from the trace identity in that lemma's proof. \par
For invariance, write the next state as the image of a single unit-ball vector. The matrix \begin{equation*} T=\left[F_KP^{1/2}/\sqrt\alpha\quad E/\sqrt{1-\alpha}\right] \end{equation*} satisfies $TT^\top=P$, so it maps the unit ball onto $\mathcal E(P)$. If $x=P^{1/2}v$ with $\|v\|\leq1$, then $F_Kx+Ew=T[\sqrt\alpha\,v^\top\ \sqrt{1-\alpha}\,w^\top]^\top$ belongs to $\mathcal E(P)$ because the stacked vector has norm at most one. This also covers singular $P$. \par
By Lemma~\ref{lem:vi}, $K_j$ approaches a strictly admissible gain. In its neighbourhood the Lyapunov solution exists, is positive definite, and depends continuously on $K_j$. Consequently the acceptance test eventually succeeds, $S^{K_j}\to S_\star$, and $U_j-L_j\to0$. \end{proof} \par
A small increment $S_{j+1}-S_j$ alone cannot replace the acceptance test. Once $K_j$ is accepted, subtraction of the two Bellman identities gives \begin{equation} S^{K_j}-S_j=\sum_{t=0}^{\infty}\alpha^{-t} (F_{K_j}^t)^\top(S_{j+1}-S_j)F_{K_j}^t. \label{eq:residual} \end{equation} Thus the closed-loop dynamics determine how an increment propagates into the remaining cost. \par
An \emph{endpoint evaluation} at $(\alpha,\eta)$ runs~\eqref{eq:vi} and the policy test until it returns a gain $K$ and bounds $L,U$ satisfying \begin{equation*} \begingroup\medmuskip=2mu\thickmuskip=2mu \rho(F_K)^2<\alpha,\; L\leq f(\alpha)\leq J(\alpha,K)=U,\; U-L\leq\eta. \endgroup\end{equation*} In the outer search, $L(b)$ and $U(b)$ denote the bounds returned at $b$. Different endpoints may stop at different iteration indices. \par
\section{Certification over the Full Parameter Domain} \label{sec:global} \begin{figure*}[t] \centering\includegraphics[width=\textwidth]{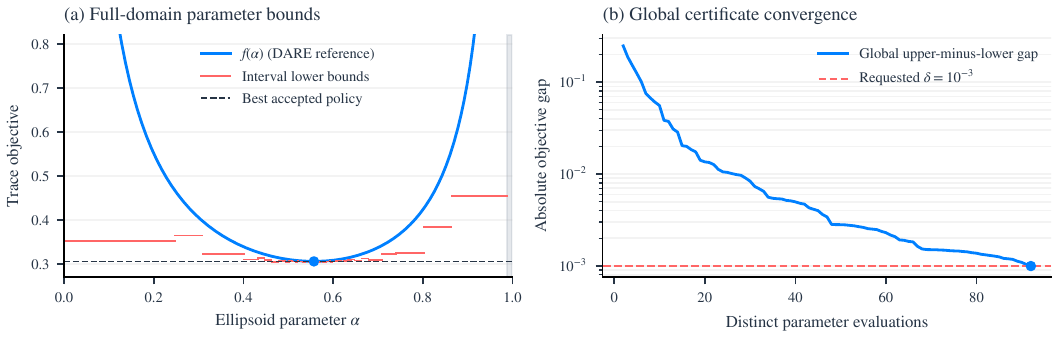} \caption{Full-domain search for~\eqref{eq:example}. Left: each red segment bounds the whole corresponding interval from below; the narrow shaded tail is excluded by~\eqref{eq:endpoint}. The model-based curve is a validation reference. Right: the global gap reaches the prescribed tolerance after $92$ parameter evaluations.} \label{fig:global} \end{figure*} \par
Bounds at a few parameter values do not exclude a better controller elsewhere. To obtain a global certificate, we bound the objective over whole intervals, including parameter values that have not been evaluated. \par
\subsection{Bounds for intervals and the upper tail} To bound an interval without sampling its interior, separate the objective into two factors that change in known directions with the parameter. By Lemma~\ref{lem:vi}, $(1-\alpha)f(\alpha)=\tr(E^\top S_\star(\alpha)E)$ is non-increasing in $\alpha$. On $[a,b]$, the right endpoint therefore bounds the numerator in~\eqref{eq:value} from below, while the left endpoint bounds the increasing factor $(1-\alpha)^{-1}$ from below. Their product bounds the objective throughout the interval. No convexity or unimodality of $f$ is needed. \par
\begin{lemma}[Interval and endpoint bounds]\label{lem:interval} For $I=[a,b]$ with $0\leq a<b<1$, any iterate at $b$ yields \begin{equation} \ell(I)=\frac{\tr(E^\top S_j(b)E)}{1-a} =\frac{1-b}{1-a}L_j(b) \leq\inf_{\alpha\in I\cap(0,1)} f(\alpha). \label{eq:interval} \end{equation} In addition, $n$ steps of~\eqref{eq:vi} at $\alpha=1$ give \begin{equation} c=\tr(E^\top S_n(1)E)>0,\qquad f(\alpha)\geq\frac{c}{1-\alpha}. \label{eq:endpoint} \end{equation} \end{lemma} \aftertheoremdisplay\begin{proof} For $\alpha\in I\cap(0,1)$, Lemma~\ref{lem:vi} gives $S_\star(\alpha)\succeq S_\star(b)\succeq S_j(b)$ and $(1-\alpha)^{-1}\geq(1-a)^{-1}$, proving~\eqref{eq:interval}. At $\alpha=1$, zero cost over $n$ stages forces all inputs and outputs to vanish, since $C^\top D=0$ and $D^\top D\succ0$. Observability then forces the initial state to vanish, so $S_n(1)\succ0$. Since $E\ne0$, $c>0$. Apply parameter monotonicity and the finite-horizon lower bound to obtain~\eqref{eq:endpoint}. \end{proof} \par
The calculation at $\alpha=1$ uses only a finite horizon to bound the tail; it requires no admissible infinite-horizon design at that endpoint. \par
An initial evaluation at $\alpha=1/2$ gives a usable policy with cost $U_{\rm s}$. Since $c/(1-\alpha)$ grows without bound near one, parameters sufficiently close to one cannot improve this cost. Set \begin{equation} b_0=1-c/U_{\rm s}\in[1/2,1). \label{eq:cutoff} \end{equation} Equation~\eqref{eq:endpoint} excludes $[b_0,1)$ from improving this policy. The remaining interval starts at zero, but only its positive right endpoints are evaluated. \par
\subsection{Refining the remaining intervals} The search keeps a partition of $[0,b_0]$ and the best controller found so far, called the incumbent. Its cost gives a global upper bound. The smallest interval bound, together with the excluded-tail bound, gives a global lower bound. Their difference measures the improvement that is still possible. We split an interval attaining the smallest lower bound and evaluate its new midpoint. Once the global bounds differ by at most $\delta$, the incumbent is sufficiently close to optimal. \par
\begin{theorem}[Finite global certificate]\label{thm:global} Under Assumption~\ref{ass:main}, Algorithm~\ref{alg:global} terminates for every $\delta>0$ and returns $(\widehat\alpha,\widehat K,\widehat P,\underline J,\overline J)$ satisfying \begin{gather} \rho(F_{\widehat K})^2<\widehat\alpha<1, \qquad\mathcal E(\widehat P)\text{ is invariant},\label{eq:guarantee1}\\ \underline J\leq J_\star\leq J(\widehat\alpha,\widehat K) =\overline J,\qquad\overline J-\underline J\leq\delta. \label{eq:guarantee2} \end{gather} No initial stabilising gain or interior minimiser is required. \end{theorem} \begin{proof} Proposition~\ref{prop:policy} makes every endpoint evaluation finite and every incumbent admissible. Lemma~\ref{lem:interval}, the excluded-tail bound $U_{\rm s}$, and the covering partition imply $\underline J\leq J_\star\leq\overline J$ throughout. \par
It remains to prove termination, including near zero. Write $\eta=\delta/4$ and consider any current interval $I=[a,b]$. Its endpoint bounds obey $U(b)-L(b)\leq\eta$ and $U(b)\geq\overline J$, since the incumbent is the smallest visited policy cost. With $q=(1-b)/(1-a)$ and $\ell(I)=qL(b)$, \begin{align} \overline J-\ell(I) &\leq\underbrace{q\eta}_{\text{endpoint error}} +\underbrace{(1-q)\overline J}_{\text{interval contribution}}\nonumber\\ &\leq\eta+\frac{b-a}{1-b_0}U_{\rm s}. \label{eq:width} \end{align} The endpoint evaluation controls the first term; bisection controls the second. A sufficiently narrow interval cannot lie more than $\delta$ below the incumbent. Hence, if the stopping test fails, the selected interval must have width greater than \begin{equation} h_\star=(\delta-\eta)(1-b_0)/U_{\rm s}>0. \label{eq:hstar} \end{equation} Each split halves the interval width. Bisection therefore cannot select any interval at depth $d_\star=\max\{0,\lceil\log_2(b_0/h_\star)\rceil\}$. There are at most $2^{d_\star}-1$ shallower nodes. The estimate applies unchanged to $[0,b]$: zero is never evaluated, and no uniform inner-iteration bound near zero is needed. Only finitely many endpoints are visited, and each evaluation terminates. At termination the stopping test gives~\eqref{eq:guarantee2}; Proposition~\ref{prop:policy} gives~\eqref{eq:guarantee1} and the returned ellipsoid. \end{proof} \par
\begin{algorithm}[t] \small\caption{Globally certified ellipsoid design} \label{alg:global} \begin{algorithmic}[1] \REQUIRE Exact batch satisfying Assumption~\ref{ass:main}; $\delta>0$. \STATE Construct~\eqref{eq:virtual}--\eqref{eq:maps}; set $\eta=\delta/4$. \STATE Evaluate $(1/2,\eta)$; store the returned policy and cost $U_{\rm s}$ as incumbent. \STATE Compute $c$ and $b_0$ from~\eqref{eq:endpoint}, \eqref{eq:cutoff}. \STATE Evaluate $(b_0,\eta)$; initialise the partition $\mathcal I=\{[0,b_0]\}$ and its bound~\eqref{eq:interval}. \LOOP\STATE Update $\overline J$ and its policy from all visited upper values; set $\underline J=\min\{U_{\rm s},\min_{I\in\mathcal I}\ell(I)\}$. \IF{$\overline J-\underline J\leq\delta$} \STATE Solve~\eqref{eq:primal} for the incumbent and return the quantities in Theorem~\ref{thm:global}. \ENDIF\STATE Bisect an interval of smallest $\ell(I)$; evaluate its midpoint with tolerance $\eta$. \STATE Replace the parent by its two children and compute~\eqref{eq:interval}, reusing the old right endpoint. \ENDLOOP\end{algorithmic} \end{algorithm} \par
\begin{figure*}[t] \centering\includegraphics[width=\textwidth]{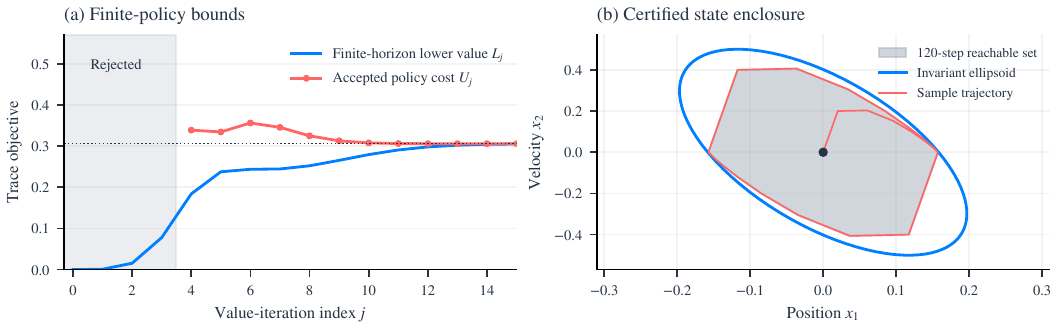} \caption{Policy and geometric checks at the returned parameter. Left: early greedy gains are rejected; accepted policy costs need not decrease at every value update. Right: the grey polygon is an inner approximation of the $120$-step reachable set from $x_0=0$, formed by joining support points in $900$ directions. The trajectory also starts at $x_0=0$ and uses $w_k=(-1)^{\lfloor k/13\rfloor}$. Containment in the ellipsoid for every admissible disturbance follows from Proposition~\ref{prop:policy}, not from this trajectory.} \label{fig:policy} \end{figure*} \par
\par
The guarantee concerns exact data and exact arithmetic. Floating-point rank decisions and Lyapunov solves require conditioning and residual checks; they do not constitute outward-rounded numerical certification or robustness to unmeasured errors. The interval construction can also be combined with a model-based Riccati solver that supplies valid endpoint bounds. \par
\section{Numerical Study} \label{sec:numerics} \subsection{A disturbed position--velocity system} Consider a sampled position--velocity system: \begin{equation} A=\begin{bmatrix}1&0.2\\0&1\end{bmatrix},\quad B=E=\begin{bmatrix}0.02\\0.2\end{bmatrix},\quad z_k=\begin{bmatrix}x_{1,k}\\0.3u_k\end{bmatrix}. \label{eq:example} \end{equation} This models position and velocity with a sampling period of $0.2$ and an unknown bounded acceleration acting through the input channel. Position is penalised; velocity is observed through the dynamics, so the state penalty is semidefinite. \par
The batch contains $20$ open-loop episodes of $8$ steps each. The initial state of each episode is drawn uniformly from $[-1,1]^2$, and inputs and measured disturbances are independent uniform samples in $[-1,1]$. Resetting after eight steps limits open-loop growth; it is an experimental capability, not a learned safety guarantee. The data matrix has rank $4$ and condition number $1.70$. Only the five measured arrays are passed to the learner; the matrices in~\eqref{eq:example} are used separately for validation. \par
Reported numerical values are rounded; all checks use the unrounded quantities. With $\delta=10^{-3}$ and $\eta=\delta/4$, the computation returns $\widehat\alpha=0.5572$ and \begin{equation} \begin{aligned} \widehat K&=\begin{bmatrix}-6.3577&-4.2834\end{bmatrix},\\ \widehat P&=\begin{bmatrix}0.0388&-0.0587\\-0.0587&0.2504\end{bmatrix}. \end{aligned} \label{eq:returned} \end{equation} The unrounded controller gives the bounds and gap \begin{equation} \begin{aligned} \underline J&=0.30464773,\quad\overline J=0.30564354,\\ \overline J-\underline J&=9.9581\times10^{-4}. \end{aligned} \label{eq:numericalgap} \end{equation} The weighted stability margin $1-\rho(F_{\widehat K})^2/\widehat\alpha$ is $0.5146$. The search uses $92$ distinct parameters, $90$ bisections, and $1569$ value updates. The normalised primal Lyapunov residual and primal--dual trace discrepancy are below $3\times10^{-16}$. These are floating-point diagnostics; the bound proved in Theorem~\ref{thm:global} is an exact-arithmetic statement. \par
Figure~\ref{fig:global} shows the interval lower bounds and the closing global gap. Figure~\ref{fig:policy} distinguishes finite-horizon critic values from accepted policy performance and displays the returned geometric enclosure. An independent DARE calculation gives the reference value $0.30564346$ near $\alpha=0.5574$. Agreement alone would not supply the lower bound in~\eqref{eq:numericalgap}. \par
\subsection{Matched computations and diagnostic cases} Table~\ref{tab:comparison} compares value iteration (VI) and Riccati (DARE) computations using the same batch and criterion. The local methods use bounded scalar minimisation on $[0.02,0.98]$, with parameter tolerance $10^{-9}$. The local VI baseline stops when \begin{equation*} \frac{\|S_{j+1}-S_j\|_F}{1-\alpha+\|S_j\|_F}\leq10^{-10}, \end{equation*} where $\|\cdot\|_F$ is the Frobenius norm. Its scalar search uses $\tr(E^\top S_{j+1}E)/(1-\alpha)$; the reported objective is obtained by evaluating the returned gain separately. The DARE global method uses the same interval bounds with exact Riccati endpoint values in theory. Numerical policy evaluation rejects Lyapunov operators with condition number above $10^{10}$. \par
Computations use Python 3.12.14, NumPy 2.5.3, and SciPy 1.18.1 on Windows 11 with an Intel Core i7-1185G7 processor at 3.00~GHz. Riccati and policy evaluations use SciPy's discrete Riccati and Lyapunov solvers. Timings are medians of three runs with one BLAS thread; preprocessing takes less than $0.3$~ms in each case. \par
\begin{table}[t] \vspace{5pt} \caption{Same-data comparison for $\delta=10^{-3}$} \label{tab:comparison} \centering\begin{tabular}{@{}lrrr@{}} \toprule Method & Evals. & Time (ms) & Global gap\\ \midrule VI / local & 12 & 12.3 & unavailable\\ DARE / local & 13 & 9.1 & unavailable\\ VI / certified & 92 & 306.8 & $9.96\!\times\!10^{-4}$\\ DARE / certified & 88 & 58.9 & $9.88\!\times\!10^{-4}$\\ \bottomrule\end{tabular} \par
\smallskip\parbox{\columnwidth}{\footnotesize Gaps are computed in floating-point arithmetic; Theorem~\ref{thm:global} assumes exact arithmetic.} \end{table} All four objectives lie within $9\times10^{-8}$ of the DARE reference. The local computations are faster but do not bound unexplored parameters. The matched global DARE computation is also faster than value iteration: the proposed contribution is the finite global certificate and its data-based realisation, not computational superiority. \par
Two scalar cases test failures that a nominal convergence plot misses. First, take $A=1.1$, $B=E=1$, $z=[10^{-7}x\ u]^\top$, and $\alpha=1/2$. The first matrix increment is $S_1-S_0=10^{-14}$, and the corresponding lower-bound increment is $L_1-L_0=2\times10^{-14}$. Either increment is below $10^{-12}$, although $K_0=0$ leaves the pole at $1.1$. In contrast,~\eqref{eq:policy} gives $S^{K_0}=10^{-14}/(1-1.1^2/\alpha)<0$, so Proposition~\ref{prop:policy} rejects the gain. Second, for the boundary example after~\eqref{eq:optimum}, tolerances $10^{-2},10^{-3},10^{-4}$ require respectively $7,10,14$ parameter evaluations. The returned parameters are $2^{-7},2^{-10},2^{-14}$ and the corresponding gaps to $J_\star=1$ are $7.874\times10^{-3}$, $9.775\times10^{-4}$, and $6.104\times10^{-5}$. No positive lower cutoff is imposed. \par
\section{Conclusion} The proposed procedure uses a finite batch of exact data collected under measured disturbances to compute a feedback gain, a scalar parameter, and their associated invariant ellipsoid, together with a finite global objective gap. Its key step is the combination of policy upper bounds with normalised-value interval lower bounds. The resulting width estimate proves termination over the entire open parameter domain and accommodates an unattained boundary infimum. The certificate also applies to model-based Riccati evaluations; its benefit is an explicit performance guarantee rather than a speed or information advantage. Extending the bounds to measurement uncertainty and verified numerical arithmetic remains a separate problem. \par
\par
\section*{Acknowledgment} \textit{AI use statement:} This letter was written with the assistance of ChatGPT 5.6 Sol. Working with AI was an iterative process, making it difficult to isolate its specific contribution. \par
 \par
\end{document}